\documentclass[journal]{IEEEtran}
\usepackage{graphicx}
\usepackage{caption}
\usepackage{subcaption}
\usepackage{stmaryrd}
\usepackage{verbatim}
\usepackage{amssymb}
\usepackage{amsmath}
\usepackage{amsfonts}
\usepackage{mathrsfs}
\usepackage{mathabx}
\usepackage{wasysym}
\usepackage{amsthm}
\usepackage{cite}
\usepackage{fixltx2e}
\usepackage{float}
\usepackage{dblfloatfix}
\usepackage[ampersand]{easylist}
\usepackage[ruled,vlined]{algorithm2e}
\usepackage{afterpage}
\usepackage{multirow}
\usepackage{tabu}
\DeclareMathAlphabet{\mathpzc}{OT1}{pzc}{m}{it}
\newcommand\norm[1]{\left\lVert#1\right\rVert}
\newcommand\abs[1]{\left\lvert#1\right\rvert}
\newcommand\inner[1]{\left\langle#1\right\rangle}
\newtheorem{lem}{Lemma}

\begin{document}
\title{Nonlinear Interference Alignment in a One-dimensional Space}
\author{Mohaned Chraiti,~Ali Ghrayeb and Chadi Assi}
\maketitle
{\let\thefootnote\relax\footnotetext{M. Chraiti and C. Assi are with the CIISE Department, Concordia University, Montreal,
Canada (email:m\_chrait@encs.concordia.ca; assi@ciise.concordia.ca).
\par A. Ghrayeb is with the ECE Department, Texas A$\&$M University at Qatar, Doha, Qatar (e-mail: ali.ghrayeb@qatar.tamu.edu).
de Montr\'eal, Montreal, Canada (e-mail: j-f.frigon@polymtl.ca).}}
\begin{abstract}
Real interference alignment is efficient in breaking-up a one-dimensional space over time-invariant channels into fractional dimensions. As such, multiple symbols can be simultaneously transmitted with fractional degrees-of-freedom (DoF). Of particular interest is when the one dimensional space is partitioned into two fractional dimensions. In such scenario, the interfering signals are confined to one sub-space and the intended signal is confined to the other sub-space. Existing real interference alignment schemes achieve near-capacity performance at high signal-to-noise ratio (SNR) for time-invariant channels. However, such techniques yield poor achievable rate at finite SNR, which is of interest from a practical point of view. In this paper, we propose a radically novel nonlinear interference alignment technique, which we refer to as Interference Dissolution (ID). ID allows to break-up a one-dimensional space into two fractional dimensions while achieving near-capacity performance for the entire SNR range. This is achieved by aligning signals by signals, as opposed to aligning signals by the channel. We introduce ID by considering a time-invariant, point-to-point multiple-input single-output (MISO) channel. This channel has a one-dimensional space and offers one DoF. We show that, by breaking-up the one dimensional space into two sub-spaces, ID achieves a rate of two symbols per channel use while providing $\frac{1}{2}$ DoF for each symbol. We analyze the performance of the proposed ID scheme in terms of the achievable rate and the symbol error rate. We also propose a decoder and prove its optimality. In characterizing the achievable rate of ID for the entire SNR range, we prove that, assuming Gaussian signals, the sum achievable rate is at most one bit away from the capacity. We present numerical examples to validate the theoretical analysis. We also compare the performance of ID in terms of the achievable rate performance to that of existing schemes and demonstrate ID's superiority.

\end{abstract}
\section{Introduction}
\par Interference is an inherent phenomenon in wireless networks, resulting from concurrent transmissions of many signals over the same communication channel. Interference can be avoided by transmitting different independent signals over orthogonal channels in the time, frequency and/or space dimensions, but this leads to poor throughput as the number of users grows large. Efficient interference management is thus crucial for efficient use of resources and for achieving high throughput. Traditionally, interference is considered as either noise or decoded along with the intended signal \cite{interference2}. Interference may thus cause severe performance degradation especially when the interference and the intended signals are of similar power strength.

\par Much work has been done to deal with interference, including the so-called interference alignment \cite{alignment}, which has opened the possibility of providing significantly better performance than what has been traditionally thought possible. The basic idea of interference alignment is to confine multiple interfering signals to a specific sub-space, while saving the other sub-spaces for communicating the intended signals. Such approach makes the impact of interference less severe at the receiver and keeps intended signals completely interference-free. This technique is especially relevant for interference channels with more than two users. Indeed, for the case where only two users share the same medium, interference at each
receiver is generated by only one user and the maximum degrees-of-freedom (DoF) offered by the channel is equal to one \cite{IntChanWithinBit}\cite{InterChanKobay}. Simple transmission schemes such as alternating access to the medium in time, for example, between the two users achieves the maximum DoF. However, such simple scheme is not efficient when the number of users is large.

\par The potential of interference alignment was first highlighted by Cadambe and Jafar \cite{alignment}, by applying it to a fully connected single-antenna $K$-user interference channel. The authors showed that, when the channel is time-varying or frequency selective, a total of $\frac{K}{2}$ DoFs (i.e., full DoF) are achievable implying that each user achieves $\frac{1}{2}$ DoF. These results are encouraging in the sense that they imply that the single-antenna $K$-user Gaussian interference channel is not inherently interference limited when the channel is time varying or frequency selective. However, the above result requires an unbounded channel diversity and it is unclear what the implication is for
real systems with finite channel diversity \cite{IntAlignFeas}.
The feasibility of interference alignment for more practical cases such as the time-invariant multiple-input multiple-output (MIMO) interference channel was studied in \cite{IntAlignFeas} with the assumption that only finite space diversity is available (no unbounded time or frequency diversity). The authors showed that, for a $K$-user interference channel, the achievable DoF is at most two for an arbitrary number of users. That is, the achievable DoF does not scale with $K$, unlike the result reported in \cite{alignment} in which the total achievable DoF is $\frac{K}{2}$. The limitation of the result in \cite{IntAlignFeas} stems from the lack of unbounded time/or frequency diversity, which is not the case in \cite{alignment}. This suggests that interference management for time-invariant channels is challenging.



\par Until recently, most of the existing interference management techniques were based on the notion of keeping some sub-spaces for interference-free communication while using successive decoding. However, since it was believed that a single DoF cannot be shared among users, non-interference-free sub-spaces essentially achieve zero DoF for at least one of the communicating signals. That is, a signal achieves either zero or one DoF. This is equivalent to saying that over a one-dimensional space (one available DoF) at most one signal belonging to the communicating signals may achieve one DoF while the remaining signals achieve zero DoF, implying that a one-dimensional space cannot be broken into fractional dimensions. Such result, for several time-invariant interference channels, is far from achieving the best performance offered by such channels.

\par It has been commonly believed that time-invariant channels are restrictive, especially when the space is one-dimensional. This perception stems from the fact that such channels are unable to incorporate vectorial interference management and one DoF is not shareable. To this end, some efforts have shown that interference alignment in a one-dimensional space is possible when symbols belong to integer lattices (lattice constellations) \cite{Lattice}. In \cite{latice1} and \cite{latice2}, the interfering symbols from several users are aligned in such a way that they are confined to one lattice. However, these techniques essentially rely on a judicious choice of power allocation such that high or low signal to interference ratio (SIR) conditions must hold, while using successive decoding.

\par The possibility of breaking up a one-dimensional space and sharing one DoF among users over time-invariant channels was first presented in \cite{one_dimension1} and \cite{one_dimension2} without imposing any conditions on the interference power. It was shown that interference management is possible if the signals are transmitted over rationally independent channels when the transmitted signals belong to a discrete constellation. They borrowed the Diophantine approximation (from number theory) to prove the possibility of communicating signals with DoF fractions. Specifically, it was shown in \cite{one_dimension2} that, at each receiver, if the channel corresponding to the interferers are rational and the channel gains corresponding to the intended signal are irrational, $\frac{K}{2}$ DoFs are achievable. In \cite{interferencGeneral} and \cite{IntAlignRobus}, more general results about this technique, referred to as real interference alignment, were presented. Specifically, it was shown in \cite{interferencGeneral} that it is possible to jointly decode interfering symbols for most channel realizations and hence they proved that the interfering symbols are naturally aligned by the channel. However, the DoF results in \cite{interferencGeneral} were derived in the asymptotic sense, i.e., at high signal-to noise ratio (SNR), which naturally raises a concern about the performance of real interference alignment at finite SNR. In \cite{IntAlignRobus}, the authors considered a $K$-user interference channel, where all interference (cross) channel gains are integers, and derived the achievable rate at finite SNR as function of some parameters to optimize over.

\par The approach of breaking up a one-dimensional space into fractional dimensions has been shown to have great potential for several applications in terms of the asymptotic achievable DoF. In \cite{one_dimension1,one_dimension2,interferencGeneral,IntAlignRobus}, the authors considered time-invariant interference X channels. They used real interference alignment to partition the space into two fractional dimensions where symbols are linearly precoded such that all the interference symbols are confined to a fractional dimension while the intended symbol is confined to the remaining one. In \cite{intAlignWiret}, this approach was also applied in the context of secrecy communication over a $K$-user Gaussian wiretap channel in the presence of $N$ eavesdroppers. The authors showed that $\frac{1}{2}-\frac{1}{2K}$ secure DoFs are achievable for almost all transmitter-receiver pairs. However, the numerical results in \cite{IntAlignRobus} showed that real interference alignment provides poor performance in terms of the achievable rate at finite SNR. Indeed, \cite{IntAlignRobus} considered the case of two interfering symbols in one-dimension and real interference alignment was used to partition the one-dimensional space into two fractional dimensions. It was shown that, up to $SNR = 40$dB, the achievable rate per symbol did not exceed $50\%$ of half the channel capacity (see Fig. 3 in \cite{IntAlignRobus}). That is, the achievable rate of both symbols (i.e., sum rate) did not exceed $50\%$ of the channel capacity, which suggests that there is much room for improvement.

\par As discussed above, real interference alignment is believed to have resulted in a big leap towards the possibility of sharing a one-dimensional space over time-invariant channels. For instance, partitioning a one-dimensional space into two fractional dimensions has been shown to be very efficient for multi-user channels and the wiretap channel. As such, symbols are linearly precoded in such a way that all interfering symbols are confined to a fractional dimension and the intended symbol gets the remaining one. However, existing real interference alignment techniques provide poor performance at finite SNR (see \cite{IntAlignRobus}). Motivated by these works, we introduce in this paper a novel approach, which we refer to as interference dissolution (ID). The proposed technique breaks up the one-dimensional space into two fractional dimensions and allows to transmit signals with fractional DoFs. The key idea is to align the signals by other signals using nonlinear precoding, rather than aligning them by the channel. To the best of our knowledge, ID is the first nonlinear interference management scheme and is the first that aligns signals by other signals. Moreover, it is the first that allows for partitioning the one-dimensional space while providing near-optimal performance for the entire SNR range in terms of the achievable rate. Therefore, we believe that ID is a radically novel way that offers the possibility of sharing a one-dimensional space by multiple signals.

\par To highlight the potential of the proposed ID, and for the sake of clarity, we adopt in this paper a generic channel model, which is a real time-invariant point-to-point MISO channel.\footnote{We remark that the proposed ID can be applied to other channels such as multi-user channels, the broadcast channel, the wiretap channel, among others. We show later in the paper how the proposed ID can be applied to a 3-user multicast channel. Other applications will be considered in future work.} The underlying channel offers a one-dimensional space and one DoF. We show how ID is able to break-up the space into two fractional dimensions and achieves fractional DoFs. This implies that the proposed ID allows for the transmission of two symbols per channel use, i.e., each symbol has $\frac{1}{2}$ DoF and thus both signals can be separated perfectly at the receiver. We stress that only real interference alignment developed for a one-dimensional space allows for simultaneous transmission of multiple symbols per channel use while providing non-zero fractional DoFs for interfering symbols \cite{one_dimension1,one_dimension2,interferencGeneral,IntAlignRobus} despite their poor performance at finite SNR, which is remedied by the proposed ID. Other interference management techniques such as successive decoding can only send symbols with either zero or one DoF, which may results in severe performance degradation, e.g., poor symbol error rate (SER).

\par In light of the above discussion, we summarize the contributions of the paper as follows:
\ListProperties(Hide=100, Hang=true, Progressive=3ex, Style*=-- ,
Style2*=$\bullet$ ,Style3*=$\circ$ ,Style4*=\tiny$\blacksquare$ )
\begin{easylist}
 & We present the proposed ID scheme including its precoding and decoding processes. We also derive the optimal decoder of the proposed scheme.
 & We show that ID allows for the transmission of $\frac{K}{\lceil\frac{K}{2}\rceil+1}\underset{K\to\infty}{\to}2$ (uncoded) symbols per channel use, where $K$ is the total number of transmitted signals. ID breaks up the one-dimensional space into two fractional spaces where each symbol achieves $\frac{\frac{K}{2}}{2(\lceil\frac{K}{2}\rceil+1)}\underset{K\to\infty}{\to}\frac{1}{2}$ DoF.
 & We characterize the achievable rate of ID at finite SNR considering Gaussian signals. We prove that the achievable rate associated with each symbol is within half a bit from half of the channel capacity. This proves that ID achieves near optimal performance at finite SNR, unlike existing real interference alignment techniques where they yield poor performance at finite SNR.
& We apply the proposed ID to a $3$-user multicast channel and show its advantages over existing techniques such as successive decoding.
& Simulation results are provided to demonstrate the efficacy of ID in separating interfering signals and thus managing interference in a one-dimensional space. Moreover, we show that ID outperforms its counterpart real interference alignment techniques \cite{one_dimension1,one_dimension2,interferencGeneral,IntAlignRobus} in terms of the achievable rate at finite SNR.
\end{easylist}
The rest of the paper is organized as follows. Section II introduces the proposed ID technique and the corresponding optimal decoder. In Section III, the achievable DoF is analyzed and the achievable rate is derived. In Section IV, we analyze the performance of ID considering a $3$-user multicast channel. Simulation results are presented in Section V. Concluding remarks are given in Section VI.
\par Throughout the rest of the paper, we use $\|\cdot\|$, $(\cdot)^{T}$ and $\langle\cdot,\cdot\rangle$ to denote the 2-norm, the transpose operators and the inner product between two vectors, respectively.
\section{Interference dissolution (ID)}\label{interference_dissolution}

\subsection{ID Precoding and Decoding}
The system model is point-to-point MISO and it is characterized by one DoF. We assume that the transmitter is equipped with $N$ antennas ($N\geq 2$) and the receiver is equipped with one antenna. The transmitter wishes to transmit $K$ independent symbols $(s_1,s_2,\ldots,s_{K})$ to the receiver. Symbols are assumed to be drawn from real-valued discrete constellations defined by $\mathcal{M}_s=A_s\mathbb{Z}^{*}_{Q_s}=\left\{-A_sQ_s,\ldots, A_s(Q_s-1), A_sQ_s\right\}$ where $A_s$ is a real constant controlling the minimum distance
between symbols \cite{PamAndQam}. Each transmitted symbol is subject to a power constraint:
$\sum_{l=-Q_s}^{Q_s}A_sl^2\leq P$.
We assume that the channel state information (CSI) is available at the transmitter and the receiver.
\par ID achieves a rate of two symbols per channel use by sending from the transmitter a combination of all the $K$ symbols during the first channel use (or time slot). Then, during subsequent channel uses, the transmitter nonlinearly precodes symbols and sends them in such a way that one pair of symbols becomes separable at the receiver. As such, the total required number of channel uses to decode the $K$ symbols is $\lceil \frac{K}{2}\rceil+1$. ID allows to transmit $\frac{K}{\lceil \frac{K}{2}\rceil+1}\underset{K\to \infty}{\rightarrow}{2}$ symbols per channel use. In this section, we describe in detail the precoding and decoding processes.

\par We note that $K$ and $N$ are independent. As such, there are two cases: $K\leq N$ and $K>N$. In the former case, the transmitter arbitrarily chooses $K$ antennas among the available $N$ antennas to transmit the $K$ symbols in the first channel use. The selection of the $K$ antennas does not affect the development of ID since our objective is find the achievable DoF for each symbol. When $K$ distinct antennas are selected, the received signal can then be written as
\begin{equation}\label{eqn3}
y_1=\sum_{k=1}^{K}h_ks_k+n_{1},
\end{equation}
where $\{h_1,h_2,\ldots,h_{K}\}$ are real channel gains, and $n_1$ is additive white Gaussian noise (AWGN) with zero mean and variance $\sigma^2$.

\par In the second case when $K>N$, the $K$ signals can be distributed arbitrarily over the available $N$ antennas without specific order. That is, it is possible for one antenna to transmit a combination of multiple independent symbols. The only restriction here is that one symbol cannot be transmitted more than once in the first channel use. In addition, at least two antennas should be used. We elaborate on this by considering the example when $N=2$. In this case, one way is to divide the $K$ symbols into two sets of symbols $\{s_1, s_2,\ldots,s_{\lfloor\frac{K}{2}\rfloor}\}$ and $\{s_{\lfloor\frac{K}{2}\rfloor+1},s_{\lfloor\frac{K}{2}\rfloor+2},\ldots, K\}$. The transmitter sends the sum of the symbols of the first set from the first antenna, and the sum of the symbols of the second set from the second antennas. That is,
\begin{equation}\label{eqn2}
y_1=\sum_{k=1}^{\lfloor\frac{K}{2}\rfloor}h_1s_k+\sum_{k=\lfloor\frac{K}{2}\rfloor+1}^{K}h_2s_k+n_{1}.
\end{equation}
Note that (\ref{eqn3}) can simplify to (\ref{eqn2}) when $h_1=h_2=\ldots=h_{\lfloor\frac{K}{2}\rfloor}$ and $h_{\lfloor\frac{K}{2}\rfloor+1}= h_{\lfloor\frac{K}{2}\rfloor+2}=\ldots=h_{K}$. Also the way the $K$ symbols are split over the two antennas is not unique. In fact, the choice of the signals sent by the same antenna does not affect the performance of the proposed technique in terms of the achievable DoF. Having said this, Equation (\ref{eqn3}) can still be applied to the case when $K>N$. Consequently, for convenience, we use (\ref{eqn3}) in the rest of the paper.

\subsection*{\underline{$(s_1,s_2)$ Precoding}}
Recall that the objective of ID is to partition the one-dimensional space into two fractional dimensions, one fractional dimension for each symbol. Consequently, ID allows to transmit two symbols per channel use, implying that precoding and decoding symbols is done pairwise. To illustrate the precoding and decoding processes, let us start by the first symbols pair $(s_1,s_2)$ (the same is done for the remaining symbol pairs as will be given later). Consequently, the remaining signals, i.e., $(s_3,s_4,\ldots,s_{K})$, are seen as interference.

\par ID precodes symbols such that the interfering signals are aligned by the signal vector $(h_2s_2,-h_1s_1)^T$ which is orthogonal to $\{h_1s_1,h_2s_2\}$. This will allow to decode $\{s_1,s_2\}$. To this end, in the second channel use, the transmitter will send a nonlinear combination of the symbols to align the remaining symbols by the symbol vector $(h_2s_2,-h_1s_1)^T$. We note here that during all channel uses, except for the first one, the transmitter communicates nonlinear combinations of symbols with the receiver. The received signal during the first channel use is a linear combination of symbols, as shown in (\ref{eqn3}), and it will be used in the decoding process of all symbol pairs as shown later. To align signals by other signals, which is the core idea of ID, we introduce the concept of signal dissolution whereby the interfering signals are forced to yield their transmit power to one of the intended signals. Considering the pair $\{s_1, s_2\}$, the transmitter dissolves the interfering signals into $s_2$ to get the following nonlinear signal form. $$h_1s_1+\beta_1 h_2s_2+n_1,$$ where $\beta_1$ is the dissolution factor. To this end, the transmitter reproduces the noiseless part of $y_1$ where it calculates the dissolution factor $\beta_1=1+\frac{\sum_{k=3}^{K}h_ks_k}{h_2s_2}$ by solving:\footnote{ In the unlikely event that $h_2$ (or any other channel gain) is close to zero, it means that the corresponding channel is in deep fade, and therefore the associated antenna may not be used, i.e., no transmission takes place from that antenna, which is what happens in practice. Moreover, the analysis in this paper pertaining to the DoF is asymptotic in the SNR, and thus there are no constraints on the transmitted power.}

\begin{equation}\label{eqn4}
h_1s_1+\beta_1 h_2s_2=h_1s_1+h_2s_2+\sum_{k=3}^{K}h_ks_k.
\end{equation}
In the second channel use, the transmitter sends a nonlinear combination of $s_1$, $s_2$ and $\beta_1$, namely, it sends $-\beta_1 s_1$ via $h_1$ and
 $s_2$ via $h_2$. As such, the received signals vector over the two first channel uses is written as
\begin{equation}\label{eqn7}
\begin{aligned}
\boldsymbol{y}=
\begin{pmatrix}
y_1 \\ y_2
\end{pmatrix}&=
\begin{pmatrix} h_1s_1 \\ h_2s_2
\end{pmatrix}+\beta_1
\begin{pmatrix}
h_2s_2 \\ -h_1s_1
\end{pmatrix}
+
\begin{pmatrix}
n_1 \\ n_2
\end{pmatrix}\\
&=\boldsymbol{v}\big(s_1,s_2\big)+\beta_1\boldsymbol{v}^{\perp}\big(s_1,s_2\big)+\boldsymbol{n}_1,
\end{aligned}
\end{equation}
where $n_2$ is AWGN with zero mean and variance $\sigma^2$. $\boldsymbol{n}_1$ is the AWGN vector $(n_1, n_2)^T$. $\boldsymbol{v}\big(s_1,\,s_2\big)$ and $\boldsymbol{v}^{\perp}\big(s_1,\,s_2\big)$ denote the vector $(h_1s_1,h_2s_2)^T$ and its orthogonal vector $(h_2s_2,-h_1s_1)^T$ respectively. One can easily see that the remaining signals are confined to the sub-space formed by the signal vector $\boldsymbol{v}^{\perp}\big(s_1,s_2\big)=(h_2s_2,-h_1s_1)^T$. Hence, the interfering signals are aligned by the intended ones. 

\begin{table*}[t]\label{table}
\begin{center}
\caption{Interference dissolution transmission process for $K$ (even) signals.}\label{table}
\begin{tabular}{|c|c|c|c|c|c|c|}
   \hline
    Channel use & $1st$ & $2nd$ & $3rd$ & $\cdots$ & $(K/2)\text{th}$ &$(K/2+1)\text{th}$ \\
 \hline
   Transmitted Signals & $\{s_1, s_2\ldots,s_K \}$ & $\{\beta_1s_1,s_2\}$ & $\{\beta_2s_3,s_4\}$ & $\cdots$ & $\{\beta_{K/2-1}s_{K-3},s_{K-2}\}$ &$\{\beta_{K/2}s_{K-1},s_{K}\}$  \\
    \hline
    Decoded signals & $\emptyset$& $\{s_1,s_2\}$ & $\{s_3,s_4\}$ & $\cdots$ & $\{s_{K-3},s_{K-2}\}$ &$\{s_{K-1},s_{K}\}$\\
       \hline
   Received signals used & $\emptyset$& $\{y_1,y_2\}$ & $\{y_1,y_3\}$ & $\cdots$ & $\{y_1,y_{k/2}\}$ &$\{y_1,y_{k/2+1}\}$\\
    \hline
\end{tabular}
\end{center}
\end{table*}
\subsection*{\underline{$(s_1,s_2)$ Decoding}}
The interfering symbols are aligned by the signals vector $(h_2s_2,-h_1s_1)^T$ as shown in (\ref{eqn7}). However, there are three unknowns, namely $s_1$, $s_2$ and $\beta_1$, while the receiver has only two signals combinations, which makes the decoding process of $s_1$, $s_2$ not trivial. Furthermore, the received signals vector is a nonlinear combination of the transmitted signals, and hence the decoder differs considerably from existing interference alignment techniques. As we show next, the proposed decoder allows to extract and without any knowledge about $\beta_1$.”
\par To decode $\{s_1,s_2\}$, the receiver starts by building the set of all possible pairs of signals $(h_1s_1,h_2s_2)$, namely, $$\mathcal{S}=\left\{\boldsymbol{v}\big(\widetilde{s}_1,\,\widetilde{s}_2\big)=\begin{pmatrix} h_1\widetilde{s}_1 \\ h_2\widetilde{s}_2 \end{pmatrix}:\, (\widetilde{s}_1,\, \widetilde{s}_2)\in \mathcal{M}_s^2\right\}.$$ Then, for each vector $\boldsymbol{v}\big(\widetilde{s}_1,\,\widetilde{s}_2\big)\in \mathcal{S}$, the decision weight component is expressed as

\small
\begin{equation}\label{eqn8}
\begin{aligned}
w(\widetilde{s}_1,\widetilde{s}_2)
&=\frac{\abs{\inner{\boldsymbol{y}-\boldsymbol{v}\big(\widetilde{s}_1,\widetilde{s}_2\big),\boldsymbol{v}\big(\widetilde{s}_1,\widetilde{s}_2\big)}}}{\norm{\boldsymbol{v}\big(\widetilde{s}_1,\widetilde{s}_2\big)}}
\\&=\frac{\abs{\inner{\boldsymbol{v}(s_1,s_2)+\beta_1\boldsymbol{v}^{\perp}(s_1,s_2)+\boldsymbol{n}-\boldsymbol{v}\big(\widetilde{s}_1,\widetilde{s}_2\big),\boldsymbol{v}\big(\widetilde{s}_1,\widetilde{s}_2\big)}}}{\norm{\boldsymbol{v}\big(\widetilde{s}_1,\widetilde{s}_2\big)}}.
\end{aligned}
\end{equation}
\normalsize
It is clear from (\ref{eqn8}) that the noiseless part of the weight component $w(\widetilde{s}_1,\widetilde{s}_2)$ is equal to zero when $(\widetilde{s}_1,\widetilde{s}_2)=(s_1,s_2)$.
Otherwise, it takes a non zero value for almost all channel realizations (see Lemma \ref{LEM1} in Subsection \ref{Dofs} and its proof in Appendix \ref{appendix}). The decision rule consists therefore of choosing the symbol vector $(\widehat{s}_1,\widehat{s}_2)$ that minimizes the weight component, i.e.,
\begin{equation}\label{eqn10}
(\widehat{s}_1,\widehat{s}_2)= \operatornamewithlimits{argmin}_{(\widetilde{s}_1,\widetilde{s}_2)\in \mathcal{M}_s^2} \, w(\widetilde{s}_1,\widetilde{s}_2)\,.
\end{equation}
The optimality of this decision rule is proved in Subsection \ref{DeciRule}. The proof that the symbols are perfectly separable using ID is provided in Section \ref{Dofs}, where we show that each symbol has $\frac{1}{2}$ DoF. We provide simulation results in Section \ref{simulation} to validate the separability of the symbols. It should be noted that the receiver is able to decode the information symbols $(s_1,s_2)$ without decoding the remaining ones and without any knowledge about them such as transmit power, modulation technique, etc., or their associated channel gains.

\subsection*{\underline{$(s_3,s_4,\ldots,s_K)$ Precoding and Decoding}}
After decoding $(s_1,s_2)$, the transmitter and the receiver consider $(s_3,s_4)$ as intended and the rest of the signals as interference. Precoding proceeds as described before, where the noiseless part of $y_1$ is used this time to dissolve $h_1s_1+h_2s_2+\sum_{k=5}^{K}h_ks_k$ in $h_3s_3$ by calculating the dissolution factor $$\beta_2=1+\frac{h_1s_1+h_2s_2+\sum_{k=5}^{K}h_ks_k}{h_3s_3}.$$ In the third channel use, the transmitter sends the nonlinearly precoded signals in order to align the interference by the intended ones. The receiver proceeds, as explained in previous part, to decode $(s_3,s_4)$ using the signals vector $(y_1,y_3)$.
 \begin{equation}\label{eqn12}
\begin{aligned}
\begin{pmatrix}
y_1 \\ y_3
\end{pmatrix}&=
\begin{pmatrix} h_3s_3 \\ h_4s_4
\end{pmatrix}+\beta_2
\begin{pmatrix}
h_4s_4 \\ -h_3s_3
\end{pmatrix}
+
\begin{pmatrix}
n_1 \\ n_3
\end{pmatrix}\\
&=\boldsymbol{v}\big(s_3,\,s_4\big)+\beta_1\boldsymbol{v}^{\perp}\big(s_3,\,s_4\big)+\boldsymbol{n_2},
\end{aligned}
\end{equation}
where $\boldsymbol{n_2}$ is the AWGN vector $(n_1,n_3)$.
The above processes of precoding and decoding can be applied to the rest of the symbol pairs. In the $\left(k+1\right)\text{th}$ ($k\in\{1,2,\ldots,\lceil\frac{K}{2}\rceil\}$) channel use, the receiver decodes the $(k)\text{th}$ signals pair.
\par After decoding the $\lceil\frac{K}{2}\rceil-1$ signals pairs, in $\lceil\frac{K}{2}\rceil$ channel uses, it still needs to decode two symbols when $K$ is even. If $K$ is odd, a symbol is chosen to form a pair with the $K\text{th}$ symbol. Otherwise, ID is applied for the $\left(\frac{K}{2}\right)\text{th}$ signals pair. The transmission process when $K$ is even is summarized in table \ref{table}. Interference dissolution allows for the transmission of $\frac{K}{\lceil\frac{K}{2}\rceil+1}\underset{K\to\infty}{\to} 2$ uncoded information symbols per channel use while keeping them separable at the receiver.

\subsection{Interference dissolution decision rule analysis}\label{DeciRule}
\par The symbols are assumed to be equally likely over the discrete constellation $\mathcal{M}_s$. Therefore, we use Maximum Likelihood (ML) \cite{DetecEstimTheory}, which is optimal in the sense of achieving the lowest SER. That is,
\begin{equation}\label{eqn15}
\begin{aligned}
 (\widehat{s}_1,\,\widehat{s}_2)&= \operatornamewithlimits{argmax}_{(\widetilde{s}_1,\widetilde{s}_2)\in \mathcal{M}^2_s } \, p\left(\boldsymbol{y}| \widetilde{s}_1,\widetilde{s}_2\right)\, \\
 &=\operatornamewithlimits{argmin}_{(\widetilde{s}_1,\widetilde{s}_2)\in \mathcal{M}^2_s} \, \sigma^2\left( \boldsymbol{y}-E[\boldsymbol{y}]\right)^T\boldsymbol{C}^{-1}_{\widetilde{s}_1,\widetilde{s}_2}\left( \boldsymbol{y}-E[\boldsymbol{y}]\right),
\end{aligned}
\end{equation}
 where $\boldsymbol{C}_{\widetilde{s}_1,\widetilde{s}_2}$ is the covariance matrix of the received signal $\boldsymbol{y}$ given that $(\widetilde{s}_1,\widetilde{s}_2)$ was transmitted. Next, we use $\eta^2$ to denote the variance of the variable $\beta_1$ given $(\widetilde{s}_1,\widetilde{s}_2)$ was transmitted, that is,

\begin{equation}\label{eqn16}
 \begin{aligned}
 \eta^2&=E[\beta_1^2]-(E[\beta_1])^2\\
 &=E\left[\left(1+\frac{\sum_{k=3}^{K}h_ks_k}{h_2\widetilde{s}_2}\right)^2\right]-\left(E\left[1+\frac{\sum_{k=3}^{K}h_ks_k}{h_2\widetilde{s}_2}\right]\right)^2\\
 &=P\frac{\sum_{k=3}^{K}(h_k)^2}{(h_2\widetilde{s}_2)^2}.
 \end{aligned}
 \end{equation}
 The inverse of the covariance matrix is written as
 \begin{equation}\label{eqn17}
\begin{aligned}
\boldsymbol{C}^{-1}_{\widetilde{s}_1,\widetilde{s}_2}&=\left(E\left[\left(\boldsymbol{y}-E[\boldsymbol{y}]\right)\left(\boldsymbol{y}-E[\boldsymbol{y}]\right)^T\right]\right)^{-1}
\\&=\left(\begin{matrix}{\eta^2} h_2^2\widetilde{s}^2_2+{\sigma}^2 & -{\eta^2}h_1h_2\widetilde{s}_1\widetilde{s}_2 \\ -{\eta^2}h_1h_2\widetilde{s}_1\widetilde{s}_2 & {\eta^2} h_1^2\widetilde{s}^2_1+{\sigma}^2\end{matrix}\right)^{-1}
\\&=\frac{\left(\begin{matrix}{\eta^2} h_1^2\widetilde{s}^2_1+{\sigma}^2 & {\eta^2}h_1h_2\widetilde{s}_1\widetilde{s}_2 \\ {\eta^2}h_1h_2\widetilde{s}_1\widetilde{s}_2 & {\eta^2} h_2^2\widetilde{s}^2_2+{\sigma}^2\end{matrix}\right)}{{\sigma}^2({\eta}^2\norm{\boldsymbol{v}(\widetilde{s}_1,\widetilde{s}_2)}^2+{\sigma}^2)}.
\end{aligned}
\end{equation}
Let $(d_1,d_2)^T$ denote the vector $\left(\boldsymbol{y}-\boldsymbol{v}(\widetilde{s}_1,\widetilde{s}_2)\right)$. In the following, we derive an explicit expression for the ML decision function in (\ref{eqn15}) to show that it is equal to the proposed weight component $w^2(\widetilde{s}_1,\widetilde{s}_2)$ and hence prove the optimality of the decoder. Considering the formula of the covariance matrix, the ML decision function can be written as

\small
\begin{equation}\label{eqn18}
\begin{aligned}
 & \frac{{\eta}^2\left\langle\boldsymbol{y}-\boldsymbol{v}(\widetilde{s}_1,\widetilde{s}_2),\boldsymbol{v}(\widetilde{s}_1,\widetilde{s}_2)\right\rangle^2+
 {\sigma}^2\|\boldsymbol{y}-\boldsymbol{v}(\widetilde{s}_1,\widetilde{s}_2)\|^2}{{{\eta}^2\norm{\boldsymbol{v}(\widetilde{s}_1,\widetilde{s}_2)}^2+{\sigma}^2}}\, \\
 &=  \frac{{\eta}^2(d_1 h_1\widetilde{s}_1+d_2h_2\widetilde{s}_2)^2+
 {\sigma}^2(d^2_1+d^2_2)}{{\eta}^2\norm{\boldsymbol{v}(\widetilde{s}_1,\widetilde{s}_2)}^2+{\sigma}^2}\,
 \\
 &=\frac{{\eta}^2\left(\norm{\boldsymbol{v}(\widetilde{s}_1,\widetilde{s}_2)}^2(d^2_1+d^2_2)-(d_1 h_2\widetilde{s}_2-d_2h_1\widetilde{s}_1)^2\right)+{\sigma}^2(d^2_1+d^2_2)}{{\eta}^2\norm{\boldsymbol{v}(\widetilde{s}_1,\widetilde{s}_2)}^2}\\
  &=\frac{(d^2_1+d^2_2)\left({\eta}^2\norm{\boldsymbol{v}(\widetilde{s}_1,\widetilde{s}_2)}^2+{\sigma}^2\right)-{\eta}^2(d_1 h_2\widetilde{s}_2-d_2h_1\widetilde{s}_1)^2}{{\eta}^2\norm{\boldsymbol{v}(\widetilde{s}_1,\widetilde{s}_2)}^2}\\
  &= d^2_1+d^2_2-\frac{{\eta}^2\left(d_1 h_2s_2-d_2h_1 s_1\right)^2
 }{{\eta}^2\norm{\boldsymbol{v}(\widetilde{s}_1,\widetilde{s}_2)}^2+{\sigma}^2}\,\\
  &\overset{a}{\simeq} d^2_1+d^2_2-\frac{{\eta}^2\left(d_1 h_2s_2-d_2h_1 s_1\right)^2
 }{{\eta}^2\norm{\boldsymbol{v}(\widetilde{s}_1,\widetilde{s}_2)}^2}\\
 &= \frac{(d^2_1+d^2_2)\norm{\boldsymbol{V}(\widetilde{s}_1,\widetilde{s}_2)}^2-\left(d_1 h_2s_2-d_2 h_1s_1\right)^2
 }{\norm{\boldsymbol{V}(\widetilde{s}_1,\widetilde{s}_2)}^2}\, \\
 &=\frac{\left(d_1h_1s_1+d_2h_2 s_2\right)^2
 }{\norm{\boldsymbol{v}(\widetilde{s}_1,\widetilde{s}_2)}^2}\, \\
 &=\frac{\left|\langle\boldsymbol{y}-\boldsymbol{v}(\widetilde{s}_1,\widetilde{s}_2),\boldsymbol{v}(\widetilde{s}_1,\widetilde{s}_2)\rangle\right|^2}{\norm{\boldsymbol{v}(\widetilde{s}_1,\widetilde{s}_2)}^2}.
\end{aligned}
\end{equation}
\normalsize
  The equality (a)  in (11) is valid for medium to high SNR where ${\eta}^2\boldsymbol{v}(\widetilde{s}_1,\widetilde{s}_2)^2+{\sigma}^2\simeq {\eta}^2\boldsymbol{v}(\widetilde{s}_1,\widetilde{s}_2)^2$ [\cite{Goldsmith}, 6.1.7]. The ML detector in (\ref{eqn15}) hence becomes

 \begin{equation}
 \begin{aligned}
(\widehat{s}_1,\,\widehat{s}_2)&=\operatornamewithlimits{argmin}_{(\widetilde{s}_1,\widetilde{s}_2)\in \mathcal{M}^2_s}\,\, \frac{\left|\langle\boldsymbol{y}-\boldsymbol{v}(\widetilde{s}_1,\widetilde{s}_2),\boldsymbol{v}(\widetilde{s}_1,\widetilde{s}_2)\rangle\right|^2}{\norm{\boldsymbol{v}(\widetilde{s}_1,\widetilde{s}_2)}^2}\\
 &=\operatornamewithlimits{argmin}_{(\widetilde{s}_1,\widetilde{s}_2)\in \mathcal{M}^2_s}\,\, w^2(\widetilde{s}_1,\widetilde{s}_2)\\
 &=\operatornamewithlimits{argmin}_{(\widetilde{s}_1,\widetilde{s}_2)\in \mathcal{M}^2_s}\,\, w(\widetilde{s}_1,\widetilde{s}_2).
 \end{aligned}
 \end{equation}
 This proves the optimality of the proposed decoder, defined in (\ref{eqn10}), at high SNR.

\section{On The Achievable DoF of ID}\label{performance}
As mentioned above, ID applies precoding to a sequence of symbols, which is different from channel coding that is in general applied to a sequence of bits with the objective of enhancing data transmission reliability. Moreover, the channel capacity of MISO is well known and can be achieved using transmit maximum ratio combining (MRC) \cite{Tse:2005:FWC:1111206}. However, traditional precoding techniques are suited for simultaneous transmission of more than one symbol per channel use while maintaining the property that the achievable rate associated with each symbol scales with the transmit power. Otherwise, at  most one symbol will have a non-zero (i.e., one) DoF and the rest will have zero DoFs. This is true even if the sum of their respective rates is much less than the channel capacity. On the other hand, ID allows for simultaneous transmission of two symbols per channel use while ensuring that each symbol has $\frac{1}{2}$ DoF. We stress here that the sum of their respective rates is less than or equal to the channel capacity.

\par In the previous section, we presented the optimal decoder for ID, and it was clear that the proposed decoder does not require to operate under asymptotic conditions, e.g., high SNR. In this section, however, we analyze the performance of ID in the asymptotic sense to make the analysis of the error probability more tractable and to get insight about the achievable DoF.  That is, we use the asymptotic results obtained from the error probability performance analysis to reveal the achievable DoF by each symbol. We show that each symbol achieves $\frac{1}{2}$ DoF, which proves that ID is capable of breaking up a one-dimensional space into two fractional dimensions \cite{interferencGeneral}. Note that it is assumed here that the input is drawn from a real, discrete constellation.


\subsection{Error Probability Performance}\label{DofAnaly}
The optimal decoder presented in (\ref{eqn10}) chooses the symbol vector minimizing the following decision weight component.
$$w(\widetilde{s}_1,\widetilde{s}_2)=\frac{\langle\boldsymbol{y}-\boldsymbol{v}(\widetilde{s}_1,\widetilde{s}_2),\boldsymbol{v}(\widetilde{s}_1,\widetilde{s}_2)\rangle^2}{\norm{\boldsymbol{v}(\widetilde{s}_1,\widetilde{s}_2)}^2}.$$ The receiver correctly decodes the signal vector $\boldsymbol{v}(s_1,s_2)$ only if $$w(s_1,s_2)< w(\widetilde{s}_1,\widetilde{s}_2)\,\,\, \forall (\widetilde{s}_1,\widetilde{s}_2)\in\left\{\mathcal{M}_1^2: (\widetilde{s}_1,\widetilde{s}_2)\neq (s_1,s_2)\right\}.$$ For notational convenience, we use $$\boldsymbol{y}(\beta_1,s_1,s_2)=\boldsymbol{v}(s_1,s_2)+\beta_1\boldsymbol{v}^{\perp}(s_1,s_2)$$ to refer to the noiseless part of $\boldsymbol{y}$ in (\ref{eqn7}). Then we have

\small
 \begin{equation}\label{eqn19}
 \begin{aligned}
w(s_1,s_2)&=\frac{\langle\boldsymbol{y}-\boldsymbol{v}(s_1,s_2),\boldsymbol{v}(s_1,s_2)\rangle^2}{\norm{\boldsymbol{v}(s_1,s_2)}^2}\\
&=\frac{\langle\boldsymbol{n},\boldsymbol{v}(s_1,s_2)\rangle^2}{\norm{\boldsymbol{v}(s_1,s_2)}^2},
 \end{aligned}
 \end{equation}
\normalsize
 and
 \small
 \begin{equation}\label{eqn20}
 \begin{aligned}
 w(\widetilde{s}_1,\widetilde{s}_2)&= \frac{\langle\boldsymbol{y}-\boldsymbol{v}(\widetilde{s}_1,\widetilde{s}_2),\boldsymbol{v}(\widetilde{s}_1,\widetilde{s}_2)\rangle^2}{\norm{\boldsymbol{v}(\widetilde{s}_1,\widetilde{s}_2)}^2}
 \\&=\frac{\langle\boldsymbol{y}(\beta_1,s_1,s_2)-\boldsymbol{v}(\widetilde{s}_1,\widetilde{s}_2)
,\boldsymbol{v}(\widetilde{s}_1,\widetilde{s}_2)\rangle^2+\langle\boldsymbol{n},\boldsymbol{v}(\widetilde{s}_1,\widetilde{s}_2)\rangle^2}{\norm{\boldsymbol{v}(\widetilde{s}_1,\widetilde{s}_2)}^2}\\
&\,\,+ \frac{2\langle\boldsymbol{y}(\beta_1,s_1,s_2)-\boldsymbol{v}(\widetilde{s}_1,\widetilde{s}_2)
,\boldsymbol{v}(\widetilde{s}_1,\widetilde{s}_2)\rangle\langle\boldsymbol{n},\boldsymbol{v}(\widetilde{s}_1,\widetilde{s}_2)\rangle}{\norm{\boldsymbol{v}(\widetilde{s}_1,\widetilde{s}_2)}^2}.
 \end{aligned}
 \end{equation}
 \normalsize
At high SNR, the term $$\frac{\langle\boldsymbol{n},\boldsymbol{v}(\widetilde{s}_1,\widetilde{s}_2)\rangle^2}{\norm{\boldsymbol{v}(\widetilde{s}_1,\widetilde{s}_2)}^2}
-\frac{\langle\boldsymbol{n},\boldsymbol{v}(s_1,s_2)\rangle^2}{\norm{\boldsymbol{v}(s_1,s_2)}^2},$$ which is in the order of $\sigma^2$, is negligible [see Section 6.1.7 in \cite{Goldsmith}] compared to the term $$\frac{2\langle\boldsymbol{y}(\beta_1,s_1,s_2)-\boldsymbol{v}(\widetilde{s}_1,\widetilde{s}_2)
,\boldsymbol{v}(\widetilde{s}_1,\widetilde{s}_2)\rangle\langle\boldsymbol{n},\boldsymbol{v}(\widetilde{s}_1,\widetilde{s}_2)\rangle}{\norm{\boldsymbol{v}(\widetilde{s}_1,\widetilde{s}_2)}^2},$$
which is in the order of $\sqrt{P}\sigma$. The inequality $w(s_1,s_2)<w(\widetilde{s}_1, \widetilde{s}_2)$ then gives
 \begin{equation}\label{eqn21}
 \begin{aligned}
 &0<\frac{2\langle\boldsymbol{y}(\beta_1,s_1,s_2)-\boldsymbol{v}(\widetilde{s}_1,\widetilde{s}_2)
,\boldsymbol{v}(\widetilde{s}_1,\widetilde{s}_2)\rangle\langle\boldsymbol{n},\boldsymbol{v}(\widetilde{s}_1,\widetilde{s}_2)\rangle}{\norm{\boldsymbol{v}(\widetilde{s}_1,\widetilde{s}_2)}^2}\\
&+\frac{\langle\boldsymbol{y}(\beta_1,s_1,s_2)-\boldsymbol{v}(\widetilde{s}_1,\widetilde{s}_2)
,\boldsymbol{v}(\widetilde{s}_1,\widetilde{s}_2)\rangle^2}{\norm{\boldsymbol{v}(\widetilde{s}_1,\widetilde{s}_2)}^2}
 \\ \Rightarrow & 2\langle\boldsymbol{y}(\beta_1,s_1,s_2)-\boldsymbol{v}(\widetilde{s}_1,\widetilde{s}_2)
,\boldsymbol{v}(\widetilde{s}_1,\widetilde{s}_2)\rangle\langle\boldsymbol{n},\boldsymbol{v}(\widetilde{s}_1,\widetilde{s}_2)\rangle >\\
&-\langle\boldsymbol{y}(\beta_1,s_1,s_2)-\boldsymbol{v}(\widetilde{s}_1,\widetilde{s}_2)
,\boldsymbol{v}(\widetilde{s}_1,\widetilde{s}_2)\rangle^2.
\end{aligned}
 \end{equation} \\
Given that $\boldsymbol{n}$ is an AWGN vector, we can obtain a lower bound on the probability of correct decision as follows.
\begin{equation}\label{eqn22}
\begin{aligned}
P_c &\geq \mathcal{Q}\left(-\operatornamewithlimits{min}_{\underset{(\widetilde{s}_1,\widetilde{s}_2)\neq (s_1,s_2)}{(\widetilde{s}_1,\widetilde{s}_2)\in\left\{\mathcal{M}_s^2\right\}}} \frac{\left|\langle\boldsymbol{y}(\beta_1,s_1,s_2)-\boldsymbol{v}(\widetilde{s}_1,\widetilde{s}_2),\boldsymbol{v}(\widetilde{s}_1,\widetilde{s}_2)
\rangle\right|}{2\norm{\boldsymbol{v}(\widetilde{s}_1,\widetilde{s}_2)}{\sigma}}\right)\\
&=\mathcal{Q}\left(-\frac{d_{min}}{2\sigma}\right),
\end{aligned}
\end{equation}
 where $\mathcal{Q}(.)$ denotes the Q-function and  $$d_{min}=\operatornamewithlimits{min}_{\underset{(\widetilde{s}_1,\widetilde{s}_2)\neq (s_1,s_2)}{(\widetilde{s}_1,\widetilde{s}_2)\in\left\{\mathcal{M}_s^2\right\}}} \frac{\left|\langle\boldsymbol{y}(\beta_1,s_1,s_2)-\boldsymbol{v}(\widetilde{s}_1,\widetilde{s}_2),\boldsymbol{v}(\widetilde{s}_1,\widetilde{s}_2)
\rangle\right|}{\norm{\boldsymbol{v}(\widetilde{s}_1,\widetilde{s}_2)}}.$$
We now use the lower bound in (\ref{eqn22}) to provide an upper bound on the error probability, $P_e$, as follows \cite{DetecEstimTheory}.
 \begin{equation}\label{eqn23}
 \begin{aligned}
 P_e&=1-P_c\\
 &\leq 1- \mathcal{Q}\left(-\frac{d_{min}}{2\sigma}\right)
 \\&= \mathcal{Q}\left(\frac{d_{min}}{2\sigma}\right)
 \\ &\leq \exp\left(-\frac{d_{\min}^2}{8\sigma^2}\right).
 \end{aligned}
 \end{equation}
\subsection{The Achievable DoF}\label{Dofs}
We provide here a lower bound on the achievable DoF associated with each symbol per channel use. The achievable DoF represents the rate of growth of the achievable rate with respect to $\frac{1}{2}\log_2(P)$ ($\log_2(P)$  for the complex case) when $P$ tends to infinity. The term $\frac{1}{2}\log_2(P)$  follows from the well known formula of the capacity for the point-to-point AWGN channel, namely $\frac{1}{2}\log_2(1+\text{SNR})$, when the SNR tends to infinity. We first derive the achievable DoF for symbols $s_1$ and $s_2$  then we generalize it for the rest of the symbols. For the rest of this paper, without loss of generality, we assume that $s_1$ and $s_2$ are transmitted over the same channel $h=h_1=h_2$  in order to simplify the analysis. This simply can be achieved by transmitting each pair of symbols from the same antenna, although it is not necessary. Note that, due to symmetry, $s_1$ and $s_2$  have the same mutual information, i,e,. the same rate, which is

\begin{equation}\label{eqn24}
\begin{aligned}
R(s_2)&=R(s_1)\\&=I(s_1;\boldsymbol{y})\\
&=H(s_1)-H(s_1|\boldsymbol{y})
\\&\overset{b}{\geq}\left(1-P_r(s_1\neq\widehat{s}_1)\right)\log_2\left(2Q_s\right)-H(P_r(s_1\neq\widehat{s}_1)),
\end{aligned}
\end{equation}
where $I(\cdot)$ and $H(\cdot)$ denote respectively the mutual information and the entropy. $P_r(\cdot)$ denotes the probability of an event. The inequality (b) follows in (\ref{eqn24}) from the Fano's inequality
\begin{equation}\label{eqn25}
H(s_1|\boldsymbol{y})\leq H(P_r(\widehat{s}_1\neq\widehat{s}_1))+P_r(\widehat{s}_1\neq\widehat{s}_1) H(s_1).
\end{equation}
Moreover, we have $$ P_r(s_1\neq\widehat{s}_1)\leq P_r((s_1, s_2)\neq(\widehat{s}_1,\widehat{s}_2))=P_e,$$
and hence replacing $P_r(s_1\neq\widehat{s}_1)$ by $P_e$ gives the following lower bound on $R(s_1)$ and $R(s_2)$.
\begin{equation}\label{eq20}
R(s_2)=R(s_1)
{\geq}\left(1-P_e\right)\log_2\left(2Q_s\right)-H(P_e).
\end{equation}
In (\ref{eq20}), we derived a lower bound on the achievable rate as a function of $P_e$  which is given as a function of $d^2_{min}$ given in (\ref{eqn23}). It remains to provide an explicit lower bound on $d^2_{min}$  in order to derive an explicit lower bound on
the achievable rate. In the following lemma, we provide an explicit expression for the lower bound on $d^2_{min}$.
\begin{lem}\label{LEM1}
For almost all channel realizations, there exists a real constant $\mathcal{K}\in \mathbb{R}^{*+}$ such that $$d^2_{min}\geq \frac{h^2A_s^2\mathcal{K}^2}{Q_s^2}.$$
\end{lem}
\begin{proof}See Appendix \ref{appendix}.\footnote{ In \cite{interferencGeneral}, the received signal is expressed as a linear combination of the transmitted signals and the lower bound on the minimum distance follows from direct application of the Khintchine-Groshev Theorem. However, using ID, the received signal is a nonlinear combination of signals and the proof differs considerably from the one given in \cite{interferencGeneral}.}\end{proof}
The minimum value of the noiseless part of $w(\widetilde{s}_1,\widetilde{s}_2)$, over $\{(\widetilde{s}_1,\widetilde{s}_2)\in\mathcal{M}_s^2:(\widetilde{s}_1,\widetilde{s}_2)\neq(s_1,s_2)\}$ equals to $d_{min}$ which is strictly greater than zero for almost all channel realizations. This proves that the noiseless part of $w(\widetilde{s}_1,\widetilde{s}_2)$ gives zero only if $(\widetilde{s}_1,\widetilde{s}_2)=(s_1,s_2)$.
\par Each transmitted signal from $\{s_1,s_2\}$ is subject to a power constraint $P$. When $P$ is relatively large, it can be expressed as

\begin{equation}\label{eqn38}
\begin{aligned}
P&=E[s_i^2]=\frac{1}{2Q_s}\overset{l=Q_s}{\underset{l=-Q_s,l\neq0}{\sum}}A_s^2l^2\\
&=\frac{A_s^2}{Q_s}\overset{l=Q_s}{\underset{l=1}{\sum}}l^2\\
&=\frac{A_s^2Q_s(Q_s+1)(2Q_s+1)}{6Q_s}\\
&=\frac{A_s^2(Q_s+1)(2Q_s+1)}{6}\\
&\simeq \frac{A_s^2Q_s^2}{3}.
\end{aligned}
\end{equation}
Lemma \ref{LEM1}, (\ref{eqn23}) and (\ref{eqn38}) give an upper bound on the error probability as follows.
\begin{equation}\label{eqn41}
P_e \leq \exp\left(-\frac{3h^2\mathcal{K}^2P}{8\sigma^2Q_s^{4}}\right).
\end{equation}
Let $\epsilon$ be an arbitrarily small constant such that $\epsilon > 0$. In order to get a small error probability as desired at high transmit power, the transmit constellation size may be selected as follows.
\begin{equation}\label{eqn42}
Q_s^4=P^{1-\epsilon}
\Rightarrow Q=P^{\frac{1-\epsilon}{4}}.
\end{equation}
Now, we use equations (\ref{eqn24}), (\ref{eqn41}) and (\ref{eqn42}) to give a lower bound on the achievable rate. That is,
\begin{equation}\label{eq44}
\begin{aligned}
R(s_1)\geq &\left(1-\exp\left(-\frac{3h^2 \mathcal{K}^2P^\epsilon}{8\sigma^2}\right)\right)\\& \times \log_2\left(2P^{\frac{1-\epsilon}{4}}\right)-H\left(\exp\left(-\frac{3h^2 \mathcal{K}^2P^\epsilon}{8\sigma^2}\right)\right).
\end{aligned}
\end{equation}
Recall that the achievable DoF represents the rate of growth of the achievable rate with respect to $\frac{1}{2}\log_2(P)$ when $P$ tends to infinity. Therefore, the DoF associated with $s_1$ and $s_2$ can be provided by using (\ref{eq44}) as follows.
\begin{equation}\label{eqn43}
\begin{aligned}
\underset{P\rightarrow \infty}{\lim} \frac{R(s_1)}{\frac{1}{2}\log_2(P)} &\geq \underset{P\rightarrow \infty}{\lim} \frac{\log_2\left(2P^{\frac{1-\epsilon}{4}}\right)}{{\frac{1}{2}\log_2(P)}}+ \frac{H(0)}{\frac{1}{2}\log_2(P)}\\
&=\underset{P\rightarrow \infty}{\lim} \frac{2(1-\epsilon)\log_2\left(P\right)}{4\log_2(P)}\\
&=\frac{1}{2}-\frac{\epsilon}{2}.
\end{aligned}
\end{equation}
Since $\epsilon$  is arbitrarily small, we can conclude that each of $s_1$ and $s_2$ gets at least $\frac{1}{2}$ DoF. The same holds true for the remaining  $K-2$ symbols. In Section \ref{interference_dissolution}, we showed that two independent signals can be transmitted per channel use when $K$ is large.
 Since the MISO channel offers one DoF, and each symbol (as per (\ref{eqn43})) achieve at least $\frac{1}{2}$ DoF, then each symbol achieves exactly $\frac{1}{2}$ DoF. This suggests that the proposed ID technique allows for transmitting up to two independent symbols per channel use and each symbol gets $\frac{1}{2}$ DoF. This proves the capability of ID to achieve two symbols per channel use while providing $\frac{1}{2}$  DoF for each symbol.
\section{On The Achievable Rate of ID}
As shown above, ID allows for the transmission of two symbols per channel use while providing $\frac{1}{2}$ DoF for each symbol, i.e., ID breaks up the one-dimensional space into two fractional dimensions. Those results were obtained under the assumption that the symbols are drawn from a real discrete constellation and the SNR tends to infinity. This suggests that each symbol gets an associated rate that asymptotically equals half of the channel capacity. It is of interest, however, to study the achievable rate of the proposed ID for the entire range of SNR. For discrete inputs, obtaining an explicit expression for the achievable rate, using Fano’s inequality, is intractable at finite SNR. Therefore, in this section, to derive an expression for the achievable rate, we assume that the input is drawn from a Gaussian distribution, and consequently the received signal tends to Gaussian.

\par Recall that the capacity of the MISO channel is well known, which can be expressed as
\begin{equation}\label{eq1}
C=\frac{1}{2}\log_2\left(1+\frac{P_s\sum_{n=1}^{N}|g_n|^2}{\sigma^2}\right),
\end{equation}
where $\{g_1,g_2, \ldots,g_N\}$ are the channel gains between the transmit antennas and the receive antenna, and $P_s$ is the transmit power per symbol. Now we investigate the achievable rate of the proposed ID and relate it to the capacity expression given by (\ref{eq1}). One of the assumptions we had from the very beginning was that $K$ tends to be large to achieve two symbols per channel use, which implies that $K>N$. Furthermore, the set  $\{g_1,g_2, \ldots,g_N\}$  contains the set $\{h_1, h_2,\ldots,h_k\}$. As such, when all antennas are used, we have
\begin{equation}\label{eq2}
\sum_{k=1}^{K}|h_k|^2>\sum_{n=1}^{N}|g_n|^2.
\end{equation}

\par During the first channel use, the transmitter sends $K$ symbols while during each of the $\left\lceil\frac{K}{2}\right\rceil$ remaining channel uses only two symbols are transmitted. The total transmit power is thus $$\left(K+2\left\lceil\frac{K}{2}\right\rceil\right)P\simeq2KP.$$
Since the number of transmitted symbols is $K$, then the transmit power per symbol is $P_s=2P$, as per (\ref{eq2}). Consequently, the capacity expression in (\ref{eq1}) can be upper bounded as

\begin{equation}\label{eqcapacity}
\begin{aligned}
C&=\frac{1}{2}\log_2\left(1+\frac{2P\sum_{n=1}^{N}|g_n|^2}{\sigma^2}\right)\\
&<\frac{1}{2}\log_2\left(1+\frac{2P\sum_{k=1}^{K}|h_k|^2}{\sigma^2}\right)\\
&<\frac{1}{2}\log_2\left(2+\frac{2P\sum_{k=1}^{K}|h_k|^2}{\sigma^2}\right)\\
&=\frac{1}{2}\left(\log_2\left(1+\frac{P\sum_{k=1}^{K}|h_k|^2}{\sigma^2}\right)+1\right).
\end{aligned}
\end{equation}
\par We next provide the achievable rate associated with symbols $s_1$ and $s_2$ defined by (\ref{eqn7}). We then generalize the result to the remaining $K-2$ symbols. For a given channel realization, the achievable rate for $s_1$ and $s_2$ is written as

\begin{equation}\label{eqn44}
\begin{aligned}
R(s_1,s_2)&=I(s_1,s_2;y_1,y_2)\\
&=H(y_1,y_2)-H(y_1,y_2|s_1,s_2)\\
&=\frac{1}{2}\log_2\left(\det\left[C(y_1,y_2)\right]\right)\\&-\frac{1}{2}\log_2\left(\det\left[C(y_1,y_2|s_1,s_2)\right]\right),
\end{aligned}
\end{equation}
 where $C(y_1,y_2)$ is the covariance of $(y_1, y_2)$. $C(y_1,y_2|s_1,s_2)$ denotes the covariance of $(y_1, y_2)$ given $(s_1, s_2)$. These two terms can be expressed as
\begin{equation}\label{eqn45}
\begin{aligned}
C(y_1,y_2)&=\begin{pmatrix}E[y_1^2]& E[y_1 y_2]\\
 E[y_1 y_2]&  E[y_2^2]\end{pmatrix}\\
 &=P\begin{pmatrix}\sum_{k=1}^{K}|h_k|^2& 0\\
0 &
\sum_{k=1}^{K}|h_k|^2\end{pmatrix}\\
 &+ \begin{pmatrix}\sigma^2 & 0\\ 0& \sigma^2\end{pmatrix},
\end{aligned}
\end{equation}
and
\begin{equation}\label{eqn46}
\begin{aligned}
C(y_1,y_2|s_1,s_2)&=\begin{pmatrix}E[y_1^2]& E[y_1 y_2]\\
 E[y_1 y_2]&  E[y_2^2]\end{pmatrix}\\
 &=P\begin{pmatrix}\sum_{k=3}^{K}|h_k|^2& -\sum_{k=3}^{K}|h_k|^2\frac{s_1}{s_2}\\
 -\sum_{k=3}^{K}|h_k|^2\frac{s_1}{s_2} & \sum_{k=3}^{K}|h_k|^2\frac{s_1^2}{s_2^2}
\end{pmatrix}\\
&+\begin{pmatrix}\sigma^2 & 0\\ 0& \sigma^2\end{pmatrix}.
\end{aligned}
\end{equation}
Using (\ref{eqn45}) and (\ref{eqn46}), we can expand the two terms in (\ref{eqn44}), respectively, as follows. (The expression corresponding to the second term in (\ref{eqn44}) is given in (\ref{eqn48}) on the next page.) The first term in (\ref{eqn44}) can be expressed as
\begin{equation}\label{eqn47}
\begin{aligned}
\log_2\left(\det\left[C(y_1,y_2)\right]\right)&=
\log_2\left(P\sum_{k=1}^{K}|h_k|^2+\sigma^2
\right)^2\\
&=2\log_2\left(P\sum_{k=1}^{K}|h_k|^2+\sigma^2
\right).
\end{aligned}
\end{equation}
\begin{figure*}
\begin{equation}\label{eqn48}
\begin{aligned}
\log_2\left(E\left[\det\left[C(y_1,y_2|s_1,s_2)\right]\right]\right)&=
\log_2\left(\left(P\sum_{k=3}^{K}|h_k|^2+\sigma^2\right)\left(P\sum_{k=3}^{K}|h_k|^2+\sigma^2\right)-\left(P\sum_{k=3}^{K}|h_k|^2\right)^2
\right)\\
&=\log_2\left(\sigma^2\left(2P\sum_{k=3}^{K}|h_k|^2+\sigma^2\right)\right)\\
&=\log_2(\sigma^2)+\log_2\left(2P\sum_{k=3}^{K}|h_k|^2+\sigma^2\right).\\
\end{aligned}
\end{equation}
\rule{\textwidth}{0.5pt}
\end{figure*}
Consequently, the achievable rate in (\ref{eqn44}) becomes
\begin{equation}\label{eqn49}
\begin{aligned}
R(s_1,s_2)= \frac{1}{2}\log_2\left(1+\frac{P\sum_{k=1}^{K}|h_k|^2}{\sigma^2}\right)\\
+\frac{1}{2}\log_2\left(\frac{\sigma^2+P\sum_{k=1}^{K}|h_k|^2}{2P\sum_{k=3}^{K}|h_k|^2+\sigma^2}\right).
\end{aligned}
\end{equation}
Similar steps can be followed to find a similar expression for the remaining signal pairs. Specifically, the achievable rate for the $m\text{th}$ signal pair is given as
\begin{equation}\label{eqn50}
\begin{aligned}
R(s_{2m-1},s_{2m})&=I(s_{2m-1},s_{2m}:y_1,y_{m+1})
\\&=\frac{1}{2}\log_2\left[\left(1+\frac{P\sum_{k=1}^{K}|h_k|^2}{\sigma^2}\right)\right.\\
&\left.\left(\frac{\sigma^2+P\sum_{k=1}^{K}|h_k|^2}{2
P\sum_{\underset{k\neq\{2m-1,2m\}}{k=1}}^{K}|h_k|^2+\sigma^2}\right)\right].
\end{aligned}
\end{equation}

Armed with the above results, we now find the overall achievable rate per channel use, $R$. Note that $\lceil\frac{K}{2}\rceil+1$  channel uses are used to transmit the $K$ signals. As such, $R$ can be expressed as shown in (\ref{eqn51}) (on the next page).
\begin{figure*}[t]
\begin{equation}\label{eqn51}
\begin{aligned}
R&=\frac{1}{\lceil\frac{K}{2}\rceil+1}\sum_{m=1}^{\lceil\frac{K}{2}\rceil}R(s_{2m-1},s_{2m})\\
&=\frac{1}{2\left(\lceil\frac{K}{2}\rceil+1\right)}\left(\left\lceil\frac{K}{2}\right\rceil\log_2\left(1+\frac{P\sum_{k=1}^{K}|h_k|^2}{\sigma^2}\right)\right.
+\left.\sum_{m=1}^{\lceil\frac{K}{2}\rceil}\log_2\left(\frac{\sigma^2+P\sum_{k=1}^{K}|h_k|^2}{2
P\sum_{\underset{k\neq\{2m-1,2m\}}{k=1}}^{K}|h_k|^2+\sigma^2}\right)\right).\\
\end{aligned}
\end{equation}
\rule{\textwidth}{0.5pt}
\end{figure*}
\par In the following, we provide a characterization of the achievable rate associated with each symbol to within half a bit from half of the channel capacity by providing a lower
bound on the second term in (\ref{eqn51}). We have the following upper bound on the denominator of the second term in (\ref{eqn51}), that is
\begin{equation}\label{eq4}
\begin{aligned}
2P\sum_{\underset{k\neq\{2m-1,2m\}}{k=1}}^{K}|h_k|^2+\sigma^2< 2
P\sum_{k=1}^{K}|h_k|^2+2\sigma^2.
\end{aligned}
\end{equation}

\begin{figure*}
\begin{equation}\label{eqn54}
\begin{aligned}
R &\geq \frac{\lceil\frac{K}{2}\rceil}{2\left(\lceil\frac{K}{2}\rceil+1\right)}\left(\log_2\left(1+\frac{P\sum_{k=1}^{K}|h_k|^2}{\sigma^2}\right)
+\sum_{m=1}^{\lceil\frac{K}{2}\rceil}\log_2\left(\frac{\sigma^2+P\sum_{k=1}^{K}|h_k|^2}{2
P\sum_{k=1}^{K}|h_k|^2+2\sigma^2}\right)\right)\\
&= \frac{\lceil\frac{K}{2}\rceil}{2\left(\lceil\frac{K}{2}\rceil+1\right)}\left(\log_2\left(1+\frac{P\sum_{k=1}^{K}|h_k|^2}{\sigma^2}\right)
+\log_2\left(\frac{1}{2}\right)\right)\\
\\&{\simeq} \frac{1}{2}\left(\log_2\left(1+\frac{P\sum_{k=1}^{K}|h_k|^2}{\sigma^2}\right)
-1\right)\\
&>C-1.
\end{aligned}
\end{equation}
\rule{\textwidth}{0.5pt}
\end{figure*}

\par In (\ref{eqn54}), on the next page, we provide a lower on the achievable rate as a function of the channel capacity. The second line of (\ref{eqn54}) is obtained by invoking (\ref{eq4}) in (\ref{eqn51}), whereas the third line follows from the assumption that $K$ is large. Moreover, we make use of the inequality in (\ref{eqcapacity}) to arrive at the last inequality in (\ref{eqn54}). This result suggests that the achievable rate per channel use is one bit away from the channel capacity. Since two symbols are transmitted per channel use, and by invoking symmetry, it is clear that each symbol is half a bit away from half of the capacity for the entire range of SNR. This is a much improved achievable rate as compared to the one achieved by the real alignment schemes proposed in \cite{one_dimension1,one_dimension2,interferencGeneral,IntAlignRobus} where the achievable rate does not exceed $50\%$ of the channel capacity at finite SNR.

\section{Application of ID to a $3$-User Multicast Channel}
In this section, we adapt the proposed ID to a $3$-user multicast channel (sometimes referred to as a $3$-user broadcast channel \cite{dualityBroadMulti}). In particular, the transmitter is equipped with a single antenna and each of the users is equipped with a single antenna. Clearly the available DoF is one, i.e., it is a one-dimensional space.  The transmitter will have to communicate simultaneously with the three users. To transmit signals without interference, three channel uses are required and this yields one symbol per channel use. Reducing the number of required channel uses leads to interference, and therefore linearly resolving all signals at all receivers is not possible. To deal with this, one would need to use interference management in a one-dimensional space such as real interference alignment, including the proposed ID. Otherwise successive decoding may be used but this offers poor performance \cite{latice1},\cite{latice2}.

\par As discussed in previous sections, traditional techniques such as successive decoding will result in symbols achieving either one or zero DoF. Given that the total available DoF is one, the transmitter can only transmit one signal with non zero DoF using such traditional techniques. Real interference alignment may be used, but as discussed at length in previous sections, it yields poor performance at finite SNR. As a remedy to this challenge, we employ the proposed ID, where it becomes possible to break-up the one-dimensional space into two fractional dimensions, and this yields better performance for the entire SNR range. Furthermore, it becomes possible to communicate three symbols per two channel uses, i.e., $\frac{3}{2}$ symbols per channel use is achieved, where one symbol achieves one DoF and each of the other two symbols achieves $\frac{1}{2}$ DoF, as will be shown below. Note that a higher DoF translates to a higher bandwidth efficiency, but this has no impact on the resulting SER performance.

\par To elaborate, we consider the scenario where the transmitter communicates three real symbols $s_1$, $s_2$ and $s_3$ , drawn from discrete constellations (could be the same or different constellations) with three users, $U_1$, $U_2$ and $U_3$, respectively. The channel gains between the transmitter and the three users are denoted by $h_1$, $h_2$ and $h_3$. Transmission is performed over two channel uses. In the first channel use, the transmitter sends

 $$s_1+s_2+\alpha s_3.$$

We note that the coefficient $\alpha$ is a real and irrational constant and is used to ensure that $\beta$ is also real and irrational (see Lemma \ref{LEM1}). It is assumed that $\alpha$ is known to all users. In choosing the value of $\alpha$, a good choice would be one that respects the constraint on the transmit power. For example, $\alpha= \frac{\sqrt{3}}{2}$ would be a good choice.

\par In the second channel use, the transmitter dissolves $\alpha s_3$ into $s_2$ by calculating the dissolution factor
$$\beta=1+\alpha\frac{s_3}{s_2}.$$
We note here that $\beta$ is real given that $\alpha$ is real.
The transmitter then sends $$s_2-\beta s_1.$$ In the two channel uses, the received signals vector by user $U_i$ can be written as
\begin{equation}\label{eq3}
\begin{pmatrix}y^{(i)}_1\\ y^{(i)}_2\end{pmatrix}=h_i\left(\begin{pmatrix}s_1\\s_2\end{pmatrix}+\beta\begin{pmatrix}s_2\\-s_1\end{pmatrix}\right)+\begin{pmatrix}n^{(i)}_1\\ n^{(i)}_2\end{pmatrix}
\end{equation}
where $\boldsymbol{n}^{(i)}=\begin{pmatrix}n^{(i)}_1,n^{(i)}_2\end{pmatrix}^T$ denotes an AWGN vector with zero mean and covariance matrix $\sigma^2I_{2\times2}$. Here, $I_{2\times2}$ denotes the identity matrix of size $2\times2$.

\par Now that all receivers have the signal given by (\ref{eq3}), they all proceed in decoding $s_1$ and $s_2$, following the ID decoding process described in Section \ref{interference_dissolution}. Given that $\beta$ is real, $s_1$ and $s_2$ each will get $\frac{1}{2}$ DoF over two channel uses (i.e., $\frac{1}{4}$ DoF per symbol per channel use). Next, we show that $s_3$ achieves one DoF (which is  by definition obtained at high transmit power). We recall here that ID achieves near-capacity performance for the entire SNR range as proved in the previous section.
\par  Each of $s_1$ and $s_2$ gets $\frac{1}{2}$ DoF . Thus, at high transmit power, the
transmitter can communicate reliably \cite{Cover} (i.e., with zero error probability) $s_1$ and $s_2$ with all receivers with a rate equal or less than half of the channel capacity, i.e., $\frac{1}{2}C$. In this case, the estimated symbols $\widehat{s}_1$ and $\widehat{s}_2$ match exactly $s_1$ and $s_2$, respectively. $U_3$ can
decode $s_1$ and $s_2$ with zero error probability. After
correctly decoding  $s_1$ and $s_2$, $U_3$ decodes $s_3$ using  $\widehat{y}_1^{(3)}$ derived from the signal received in the first channel use.
\begin{equation}
\begin{aligned}
\widehat{y}^{(3)}_1&=h_3\left((s_1-\widehat{s_1})+\right(s_2-\widehat{s_2})+\alpha s_3)+n^{(3)}_1\\
&=\alpha h_3s_3+n^{(3)},
\end{aligned}
\end{equation}
where $n^{(3)}$ denotes AWGN with zero mean and variance $\sigma^2$. It is clear that $s_3$ gets one DoF. In two channel uses, ID allows to communicate $s_1$, $s_2$ and $s_3$ with $U_1$, $U_2$ and $U_3$ respectively while providing $\frac{1}{2}$ DoF for each of $s_1$ and $s_2$ and one DoF for $s_3$.

\par We stress here that the decoding process described above should not be confused with successive decoding, although they may appear to be similar since $s_1$ and $s_2$ are decoded first and then $s_3$ is decoded next. The difference is explained as follows. In decoding $s_1$ and $s_2$, using ID, the interference caused by $s_3$ is removed, owing to the fact that $s_3$ is aligned by $s_1$ and $s_2$. On the hand, in successive decoding, $s_3$ causes interference while decoding $s_1$ and $s_2$. The impact of this on the performance is significant. To elaborate, let us consider $U_3$. If it first decodes $s_3$, it will consider $s_1$ and $s_2$  as noise and thus the achievable rate associated with $s_3$  does not scale with the transmit power., which means that the DoF associated with $s_3$ is zero. If $U_3$ decodes $s_3$  after decoding $s_1$ and $s_2$, it considers $s_3$  as noise in decoding $s_1$ and $s_2$. Therefore, the rate associated with $s_1$ and $s_2$ does not scale with the transmit power. The transmitter then should transmit $s_1$ and $s_2$ at a very low rate in order to allow $U_3$  to reliably decode $s_1$ and $s_2$  and achieve non zero DoF associated with $s_3$ . In this case, $s_1$ and $s_2$  will both get zero DoF, and this will result in achieving zero DoF for $U_1$ and $U_2$.

\par In conclusion, we demonstrated through the above example that ID can save one channel use. In general, since the proposed ID is able to break-up the one dimensional space into two fractional dimensions, where each user achieves $\frac{1}{2}$ DoF, ID reduces the required number of channel uses by $50\%$. i.e., each channel use can serve up to two users.

\section{Simulation Results}\label{simulation}
We performed Monte Carlo simulations to evaluate the performance of ID in terms
of the achievable rate and SER when the input signals belong to discrete constellations. We consider a $2\times 1$ MISO channel. We recall that the previous theoretical results pertaining to the ID performance are provided for a given channel realization. In this section, we consider that the channel coefficients follow the Rayleigh distribution with variance $\gamma=1$. The SER and achievable rate are averaged over a large number of channel realizations. Thus, the simulation results represent the average SER and the ergodic achievable rate which we refer to hereafter by SER and achievable rate, respectively. We consider AWGN with zero mean and variance $\sigma^2= 1$. 

\subsection{Symbol Error rate}
In this subsection, we provide the SER as a function of the transmit power. We consider $4-\text{PAM}$ for the symbol constellations. Moreover, we set the total transmission power to $2P$ per channel use. Given that ID allows the transmission of two symbols per channel use, in fairness, we compare the performance of ID to the successive decoding technique where the transmission rate is two symbols per channel use. We show first the impact of interference at the receiver considering successive decoding \cite{MRC}. The performance of our proposed technique is also compared to the $2\times1$ point-to-point MISO channel where the transmitter sends one symbol per channel use. 
\begin{figure}
\begin{centering}
\includegraphics[scale=0.45]{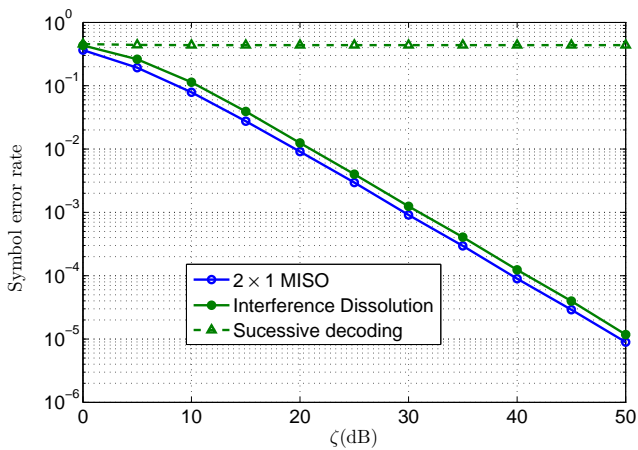}
\caption{SER versus $\zeta$(dB) for $4-\text{PAM}$ constellation points.}
\label{SymErrorRate}
\end{centering}
\end{figure}

\par  Fig. \ref{SymErrorRate} shows the SER versus $\zeta$ in dB, where $\zeta\triangleq\frac{P}{\sigma^2}$. As can be observed from Fig. \ref{SymErrorRate}, successive decoding achieves poor performance for the entire SNR range, which proves its inefficiency especially when all transmitted signals are of similar strength. The figure also shows that ID provides performance close to that of the $2\times1$ point-to-point MISO channel, which demonstrates the efficacy of ID in managing interference.
\subsection{Normalized Achievable Rate}

Fig. \ref{AchievableRate} depicts the ID achievable rate per symbol considering the lower bound in (\ref{eqn54}) for the Gaussian case and based on the Fano's inequality given in (\ref{eqn24}) for the discrete signals case. The achievable rate per symbol is normalized over the half of the MISO channel capacity $$\frac{1}{2}C=\frac{\frac{1}{2}\log_2\left(1+\frac{2P\sum_{n=1}^{N}|g_n|^2}{\sigma^2}\right)}{2}$$ given in (\ref{eqcapacity}). In the same figure, we plot the lower bound that is based on Fano's inequality, assuming discrete inputs. This lower bound is valid for both ID and real interference alignment in general.
\par Recall that the achievable rate has been shown to be at most one bit a way from the channel capacity. When the channel capacity is small, i.e., small SNR, the one-bit gab becomes more pronounced, as can be seen in Fig. \ref{AchievableRate}, where the normalized ID normalized achievable rate is close to zero at $SNR=0\text{dB}$. However, as the capacity increases, i.e., at high SNR, the one-bit gap diminishes and becomes insignificant. The achievable rate in this case approaches capacity, as is evident from Fig. \ref{AchievableRate} at $SNR=18$dB

\begin{figure}
\begin{centering}
\includegraphics[scale=0.5]{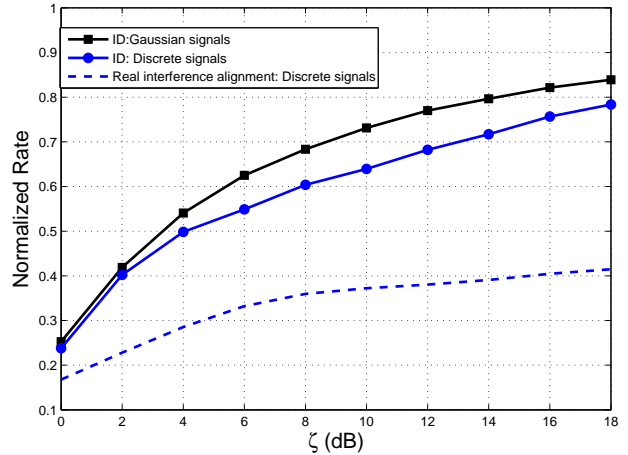}
\caption{Normalized rate as a function of $\zeta$(dB).}
\label{AchievableRate}
\end{centering}
\end{figure}



\section{Conclusion}
We proposed in this paper a nonlinear interference alignment scheme that involves breaking-up a one-dimensional space over time-invariant channels into two fractional dimensions, one used to confine the interfering signals and the other to confine the intended signal. This partition leads to achieving two symbols per channel use. The main feature of the proposed technique is that signals are aligned by other signals, which distinguishes it from existing real interference alignment schemes whereby signals are aligned by the channel. Assuming a point-to-point MISO channel, we analyzed the performance of the proposed scheme in terms of the achievable rate and DoF. We showed that the proposed technique achieves near-capacity performance and the achievable rate is always within one bit away from the channel capacity. Furthermore, since rate two symbols per channel use is achieved, we showed that each of the two symbols achieves $\frac{1}{2}$ DoF, which means that both symbols can be separated perfectly at the receiver. The proposed scheme can be applied to other scenarios with favorable results, including the many-to-one channel, multi-user channels and the wiretap channel, among others.  In fact, we adapted the proposed scheme to a $3$-user multi-cast channels and showed that rate $\frac{3}{2}$ symbols per channel use is possible to achieve. We stress here that the precoding process described in this paper applies in the general sense to other scenarios, but the specifics of the precoding process could vary from one application to another, depending on the channel model. For instance, the precoding processes for the point-to-point MISO and the $3$-user multi-cast cases considered in this paper were in general similar, but were not exactly the same.

\appendices
\section{}\label{appendix}
In this section, we use $\begin{pmatrix}\overline{d}_1\\ \overline{d}_2\end{pmatrix}$ to denote the vector  $\boldsymbol{y}(\beta_1,s_1,s_2)-\boldsymbol{v}(\widetilde{s}_1,\widetilde{s}_2)$. $d^2_{min}$ can then be written as
\begin{equation}\label{eqn27}
\begin{aligned}
d_{min}^2&= \operatornamewithlimits{min}_{\underset{(\widetilde{s}_1,\widetilde{s}_2)\neq (s_1,s_2)}{(\widetilde{s}_1,\widetilde{s}_2)\in\left\{\mathcal{M}_s^2\right\}}} \frac{|\langle(\overline{d}_1\overline{d}_2),\boldsymbol{v}(\widetilde{s}_1,\widetilde{s}_2)\rangle|^2}{\norm{\boldsymbol{v}(\widetilde{s}_1,\widetilde{s}_2)}^2}
\\&=\operatornamewithlimits{min}_{\underset{(\widetilde{s}_1,\widetilde{s}_2)\neq (s_1,s_2)}{(\widetilde{s}_1,\widetilde{s}_2)\in\left\{\mathcal{M}_s^2\right\}}} \frac{(\overline{d}_1 h \widetilde{s}_1+\overline{d}_2 h \widetilde{s}_2)^2}{(h\widetilde{s}_1)^2+(h \widetilde{s}_2)^2} \\
&=\operatornamewithlimits{min}_{\underset{(\widetilde{s}_1,\widetilde{s}_2)\neq (s_1,s_2)}{(\widetilde{s}_1,\widetilde{s}_2)\in\left\{\mathcal{M}_s^2\right\}}} \frac{\left(\overline{d}_1+\overline{d}_2 \frac{\widetilde{s}_2}{\widetilde{s}_1}\right)^2}{1+\frac{(\widetilde{s}_2)^2}{(\widetilde{s}_1)^2}}.
\end{aligned}
\end{equation}
Without loss of generality, we consider the case when $|\widetilde{s}_1|\geq |\widetilde{s}_2|$ to derive a lower bound on $d^2_{min}$. For the case $|\widetilde{s}_1|\leq |\widetilde{s_2}|$, one may follow the same steps to obtain a similar expression. At high transmit power, i.e., high value of $Q_s$, three cases are possible, namely,
\begin{equation}\label{eqn28}
\operatornamewithlimits{\lim}_{Q_s\to \infty}\frac{\widetilde{s_2}}{\widetilde{s}_1}=
\begin{cases}
\pm 1 &\mbox{if} \, |\widetilde{s}_1|\simeq|\widetilde{s_2}|\\
0 &\mbox{if} \, |\widetilde{s}_1|\gg |\widetilde{s_2}| \\
c & \mbox{otherwise}
\end{cases},
\end{equation}
where $c$ is a constant ($1 > |c| \gg \frac{1}{Q}\simeq 0$). Next, we differentiate between the three cases and we derive the desired results for each case separately. In the following, we use Khintchine-Groshev theorem. From \cite{interferencGeneral}, for almost all m-tuple real numbers $\{\alpha_1, \alpha_2,..., \alpha_m\}$ there exists a real constant $\mathcal{K}>0$ such that
\begin{equation}\label{eqn30}
|p+\alpha_1 q_1+\alpha_2 q_2+...+\alpha_m q_m|>\frac{\mathcal{K}}{(\max_i|q_i|)^{m}},
\end{equation}
for all $p\in \mathbb{Z}$ and $\{q_1, q_2,..., q_m\}\in (\mathbb{Z}^{*})^m$. So, it is important to note that $\frac{1}{A_s}\{s_1,\,s_2,\,\widetilde{s}_1,\, \widetilde{s}_2\}$ is a set of integers in $(\mathbb{Z}^{*}_{Q})^4$. In Section \ref{interference_dissolution}, we assumed that at least two signals have two different associated channel gains. This condition is important to guarantee that $\beta_1$ is a real number. Indeed, assuming that all signals are transmitted over channel $h$ only the $m\text{th}$ signal is transmitted over a channel $h_m$. In this case, $\beta_1=1+\sum_{\underset{k\neq m}{k=3}}^{K}\frac{s_k}{s_2}+\frac{h_ms_m}{hs_2}$. Now, it is clear that $\beta_1$ is a number given that $h_m$ and $h$ are two independent real channel gains \cite{interferencGeneral}.
\paragraph{Case $I$ ($\frac{\widetilde{s}_2}{\widetilde{s}_1}\simeq \pm 1$)}In this case, (\ref{eqn27}) becomes
\begin{equation}\label{eqn29}
\begin{aligned}
d^2_{min}
&\simeq \operatornamewithlimits{min}_{\underset{(\widetilde{s}_1,\widetilde{s}_2)\neq (s_1,s_2)}{(\widetilde{s}_1,\widetilde{s}_2)\in\left\{\mathcal{M}_s^2\right\}}}\frac{\left(\overline{d}_1\pm \overline{d}_2 \right)^2}{2}\\
&=\operatornamewithlimits{min}_{\underset{(\widetilde{s}_1,\widetilde{s}_2)\neq (s_1,s_2)}{(\widetilde{s}_1,\widetilde{s}_2)\in\left\{\mathcal{M}_s^2\right\}}}\frac{h^2A_s^2}{2}\left(\frac{s_1}{A_s}+\frac{s_2}{A_s}\pm\left(\frac{\widetilde{s}_1}{A_s}+\frac{\widetilde{s}_2}{A_s}\right)\right.
\\&\,\,\,\,\,+
\left.\beta_1\left(\frac{s_2}{A_s}-\frac{s_1}{A_s}\right)\right)^2\\
&\overset{c}{\geq}\operatornamewithlimits{min}_{\underset{(\widetilde{s}_1,\widetilde{s}_2)\neq (s_1,s_2)}{(\widetilde{s}_1,\widetilde{s}_2)\in\left\{\mathcal{M}_s^2\right\}}} \frac{h^2 A_s^2}{2}\left(\frac{\mathcal{K}_1}{2Q_s}\right)^2
\\&= \frac{h^2 A_s^2 \mathcal{K}_1^2}{2^{3}Q_s^{2}},\\
\end{aligned}
\end{equation}
where $\mathcal{K}_1$ is a constant in $\mathbb{R}^{*+}$. The inequality (c) is follows from up by the Khintchine-Groshev theorem.
\paragraph{Case $II$ ($\frac{\widetilde{s}_2}{\widetilde{s}_1}\simeq 0$)} In this case, (\ref{eqn27}) becomes
\begin{equation}\label{eqn31}
\begin{aligned}
d^2_{min}&\simeq\operatornamewithlimits{min}_{\underset{(\widetilde{s}_1,\widetilde{s}_2)\neq (s_1,s_2)}{(\widetilde{s}_1,\widetilde{s}_2)\in\left\{\mathcal{M}_s^2\right\}}} \overline{d}^2_1\\
&=\operatornamewithlimits{min}_{\underset{(\widetilde{s}_1,\widetilde{s}_2)\neq (s_1,s_2)}{(\widetilde{s}_1,\widetilde{s}_2)\in\left\{\mathcal{M}_s^2\right\}}} h^2 A_s^2\left(\frac{s_1}{A_s}-\frac{s_2}{A_2}
+\beta_1(\frac{s_2}{A_s})\right)^2\\
&\geq\operatornamewithlimits{min}_{\underset{(\widetilde{s}_1,\widetilde{s}_2)\neq (s_1,s_2)}{(\widetilde{s}_1,\widetilde{s}_2)\in\left\{\mathcal{M}_s^2\right\}}} h^2 A_s^2\left(\frac{\mathcal{K}_2}{Q_s}\right)^2\\ &= \frac{h^2 A_s^2 \mathcal{K}_2^2}{Q_s^{2}},
\end{aligned}
\end{equation}
where $\mathcal{K}_2\in \mathbb{R}^{*+}$. We also use the Khintchine-Groshev theorem to provide the lower bound.
\paragraph{Case $III$ ($\frac{\widetilde{s}_2}{\widetilde{s}_2}=c:0<\lim_{Q_s \to \infty}|c|<1$)} From (\ref{eqn22}), we have
\begin{equation}\label{eqn32}
d^2_{min}=\operatornamewithlimits{min}_{\underset{(\widetilde{s}_1,\widetilde{s}_2)\neq (s_1,s_2)}{(\widetilde{s}_1,\widetilde{s}_2)\in\left\{\mathcal{M}_s^2\right\}}}\frac{|\langle\left(\overline{d}_1,\overline{d}_2\right),(h\widetilde{s}_1, h\widetilde{s}_2)\rangle|^2}{h^2\widetilde{s}_1^2+h^2\widetilde{s}_2^2}.
\end{equation}
Given that $\{\widetilde{s}_1, \widetilde{s}_2, d_1, d_2\}\neq {0,0,0,0}$, the minimum is achieved when the vector $(\overline{d}_1, \overline{d}_2)$ is as orthogonal as possible to $(h\widetilde{s}_1,h\widetilde{s}_2)$, i.e., $(\overline{d}_1, \overline{d}_2)$ is the closet possible to $(h\widetilde{s}_2,-h\widetilde{s}_1)$. From the Kintchine-Groshev theorem, there are two constants ${\mathcal{K}^{(1)}_3,\mathcal{K}^{(2)}_3}$ in $\left(\mathbb{R}^{*+}\right)^2$ such that
\begin{eqnarray}
|\overline{d}_1-h\widetilde{s}_2|=|h| A_s\left|\frac{s_1}{A_s}-\frac{\widetilde{s}_2}{A_s}
+\beta_1(\frac{s_2}{A_s})\right|>|h| A_s \frac{\mathcal{K}^{(1)}_3}{Q_s}\label{eqn33},\\
|\overline{d}_2-(-h\widetilde{s}_1)|=|h| A_s\left|\frac{\widetilde{s}_2}{A_s}+\frac{s_1}{A_s}
-\beta_1(\frac{s_1}{A_s})\right|>|h| A_s \frac{\mathcal{K}^{(2)}_3}{Q_s}\label{eqn34}.
\end{eqnarray}
Next, $\mathcal{K}_3$ refers to $\min(\mathcal{K}^{(1)}_3,\mathcal{K}^{(2)}_3)$. Subsisting (\ref{eqn33}) and (\ref{eqn34}) in (\ref{eqn27}) yields
\begin{equation}\label{eqn35}
\begin{aligned}
d_{min}^2 &=\operatornamewithlimits{min}_{\underset{(\widetilde{s}_1,\widetilde{s}_2)\neq (s_1,s_2)}{(\widetilde{s}_1,\widetilde{s}_2)\in\left\{\mathcal{M}_s^2\right\}}}\frac{\left(\overline{d}_1 \widetilde{s}_1+\overline{d}_2 \widetilde{s}_2\right)^2}{\widetilde{s}_1^2+\widetilde{s}_2^2}\\
&\geq\operatornamewithlimits{min}_{\underset{(\widetilde{s}_1,\widetilde{s}_2)\neq (s_1,s_2)}{(\widetilde{s}_1,\widetilde{s}_2)\in\left\{\mathcal{M}_s^2\right\}}}\frac{\left(\left(h\widetilde{s}_2\pm h A_s \frac{\mathcal{K}_3}{Q_s}\right) \widetilde{s}_1+\left(-h\widetilde{s}_1\pm h A_s \frac{\mathcal{K}_3}{Q_s}\right) \widetilde{s}_2\right)^2}{\widetilde{s}_1^2+\widetilde{s}_2^2}\\
&=\frac{h^2 A_s^2\mathcal{K}^2_3 }{Q_s^{2}}\operatornamewithlimits{min}_{\underset{(\widetilde{s}_1,\widetilde{s}_2)\neq (s_1,s_2)}{(\widetilde{s}_1,\widetilde{s}_2)\in\left\{\mathcal{M}_s^2\right\}}}\frac{\left(\widetilde{s}_1\pm \widetilde{s}_2\right)^2}{\widetilde{s}_1^2+\widetilde{s}_2^2}\\
&\geq\frac{h^2 A_s^2 \mathcal{K}_3^2}{Q_s^{2}}\operatornamewithlimits{min}_{\underset{0<\lim_{Q_s \to \infty}|c|<1}{c\in\mathbb{Q}}}\frac{\left(1\pm c\right)^2}{1+c^2}
\\&=\frac{h^2 A_s^2 \mathcal{K}_3^2}{Q_s^{2}}c^2_{min},
\end{aligned}
\end{equation}
 where $c_{min}$ is a constant independent of $Q_s$.
\par Considering the inequality (\ref{eqn29}), (\ref{eqn31}) and (\ref{eqn35}), we obtain a lower bound on $d_{min}^2$ as follows.
\begin{equation}\label{eq35}
d_{min}^2\geq \frac{h^2 A_s^2}{Q_s^{2}} \min\left(\frac{\mathcal{K}^2_1}{2^3},\,\mathcal{K}^2_1,\,\mathcal{K}_3^2c^2_{min} \right)= \frac{h^2 A_s^2}{Q_s^{2}}\mathcal{K}^2,
\end{equation}
where $\mathcal{K}=\min\left(\frac{\mathcal{K}^2_1}{2^3},\,\mathcal{K}^2_2,\,\mathcal{K}_3^2c^2_{min}\right)$. This proves Lemma \ref{LEM1}.
\bibliographystyle{IEEEtran}
\bibliography{IEEEabrv,interferencebib}
\end{document}